\documentclass{article}

\usepackage{arxiv}

\usepackage[utf8]{inputenc} 
\usepackage[T1]{fontenc}    
\usepackage{hyperref}       
\usepackage{url}            
\usepackage{booktabs}       
\usepackage{amsfonts}       
\usepackage{nicefrac}       
\usepackage{microtype}      
\usepackage{lipsum}

\usepackage{graphicx}
\usepackage{amsthm}
\usepackage{amssymb}
\usepackage{amsmath}
\usepackage{mathtools}

\newtheorem{definition}{Definition}
\newtheorem{theorem}{Theorem}
\newtheorem{proposition}{Proposition}

\DeclarePairedDelimiter{\ceil}{\lceil}{\rceil}
\DeclarePairedDelimiter{\floor}{\lfloor}{\rfloor}

\begin{document}

\title{Gini Index based Initial Coin Offering Mechanism}


\author{
  Mingyu Guo, Zhenghui Wang\\
  School of Computer Science\\
  University of Adelaide\\
  Australia\\
   \And
   Yuko Sakurai\\
   National Institute of Advanced\\
   Industrial Science and Technology\\
   Japan\\
}

\maketitle

\begin{abstract} As a fundraising method, Initial Coin Offering (ICO) has
    raised billions of dollars for thousands of startups in the past two years.
    Existing ICO mechanisms place more emphasis on the short-term benefits of
    maximal fundraising while ignoring the problem of unbalanced token
    allocation, which negatively impacts subsequent fundraising and has bad
    effects on introducing new investors and resources. We propose a new ICO
    mechanism which uses the concept of Gini index for the very first time as a
    mechanism design constraint to control allocation inequality.  Our
    mechanism maintains an elegant and straightforward structure.  It allows
    the agents to modify their bids as a price discovery process, while
    limiting the bids of whales. We analyze the agents' equilibrium behaviors
    under our mechanism. Under natural technical assumptions, we show that most
    agents have simple dominant strategies and the equilibrium revenue
    approaches the optimal revenue asymptotically in the number of agents.  We
    verify our mechanism using real ICO dataset we collected, and confirm that
    our mechanism performs well in terms of both allocation fairness and
    revenue.  \end{abstract}

\keywords{Mechanism Design; Initial Coin Offering; Gini Index}

\section{Introduction}\label{sec:intro} As a primary fundraising tool for
startups, Initial Coin Offering (ICO) is very eye-catching in the capital
market. According to~\cite{Icobench2018:ICO}, $718$ ICOs ended in 2017 and
raised around $10$ billions (USD) in total. By 2018, the ICO fundraising market
has grown significantly. The total fundraising for the year reached $11.5$
billions (USD), and the number of ended ICOs increased to $2517$.

Despite the popularity and the huge amounts of funds being raised, popular ICO
mechanisms are surprisingly unsophisticated.  Vitalik
Buterin~\cite{Buterin2017:Analyzing} analyzed some token sale models and
claimed that an optimal token sale model has not been discovered yet.  There
are two commonly used ICO mechanisms. One is simply the fixed price mechanism.
There are two variants of this mechanism: uncapped sales and capped sales.
Uncapped sales does not limit the number of coins sold. The aim is to accept as
much capital and as many investors into the project as possible. Capped sales
sells a fixed number of coins, and cuts off additional investment once all
coins are sold.  The fixed price mechanism does not offer a price discovery
process. Many cryptocurrencies released via ICOs have different degrees of appreciation after their tokens
are listed on exchanges.  The average ICO is under priced by
$8.2\%$~\cite{Momtaz2018:Initial}, while some tokens are appreciated by more
than $1000\%$~\cite{Sanchez2017:Optimal}.  Another commonly used mechanism is
the Dutch auction, which offers a much better price discovery process. It
starts the auction with a high price, which keeps decreasing until enough
participants are willing to purchase all coins according to the market price.
The Dutch auction also has its drawbacks. \emph{E.g.}, it requires the
participants actively monitor the auction progress.


For both the fixed price mechanism with capped sales and the Dutch auction,
\emph{whales} (large investors) can end the auction early by putting in large
investments, which takes away the investment opportunities from the smaller
investors.  Wealth inequality is a significant problem in the cryptocurrency
community, which goes against the principle of decentralization, especially if
the coin is based on proof-of-stake (that favors richer holders).  It is also
worth mentioning that the tokens sold in an ICO normally amount to no more than
$50\%$ of the tokens. The unsold tokens are still in the hands of the release
team. The hidden danger of sharp depreciation (dumping large amounts of coins)
caused by whales is a vital threat for the development team.

Roubini~\cite{Roubini2018:Testimony} testified in a hearing of the US Senate
Committee on Banking, Housing and Community Affairs: ``\textit{wealth in
crypto-land is more concentrated than in North Korea where the inequality Gini
coefficient is $0.86$ (it is $0.41$ in the quite unequal US): the Gini
coefficient for Bitcoin is an astonishing $0.88$.}''  Many cryptocurrency
buyers expressed dissatisfaction to the wealth gap on Internet forums,
calling ``whale-sale completely against the ethos of Ethereum and
Cryptocurrency''~\cite{Wit222019:BAT}.  We collected the data from a few ICOs.
Based on our calculation, the Gini index of token allocation in most of the
ICOs is higher than $0.88$.  For example, for Gnosis (with a Gini index of $0.92$), the top $5\%$ of the users
invested $78\%$. In sharp contrast, the bottom $64\%$ of the users invested only
$4\%$ in total! The top two investors invested $31\%$ and $15\%$, respectively.

As a means of balance, in this paper, we introduce a Gini index based ICO
mechanism to achieve a more balanced token allocation by limiting whales.
Since Corrado Gini proposed the concept of Gini index for the first time in
1912, Gini index has been widely used to measure wealth gap. Gini index was
also studied recently in mechanism design~\cite{Sinha2015:Mechanism}. The
authors used Gini index to evaluate allocation fairness in simulations. In our
paper, the Gini index is a mechanism design constraint. Our analysis is based
on the specific mathematical structure of the Gini index.

The Gini index may not be meaningful in environments where false name
bids~\cite{Yokoo2004:effect} are possible, as a whale may simply divide her
investment and participate via multiple accounts.  Existing ICOs in practise do
often involve the \emph{Know Your Customer (KYC)} process and only the users in the \emph{whitelist} can join the ICO. KYC is a process in
which a business verifies its customers' identities, as well as the risk of
illegal intent. KYC processes are also being used to comply with anti-bribery
and other regulations (\emph{e.g.}, it is illegal to
sell cryptocurrencies to customers from countries/territories where
cryptocurrencies are banned).  The US Securities and Exchange Commission (SEC)
has announced the intent to pursue ICOs without a proper KYC processes. 
Besides document
verification, existing KYC processes use manual approaches like video
conversation to verify customers' identities.



We propose the \emph{Gini mechanism}, which has a list of desirable properties:

\begin{itemize}

    \item The Gini mechanism is a simple prior-free mechanism with elegant
        description.  It is a key advantage for practical applications.
        Data related to ICO events are extremely volatile. It is unrealistic to
        assume a known prior distribution. Therefore, computational approaches
        (\emph{e.g.}, deep learning based mechanism
        design~\cite{Duetting2019:Optimal,Shen2019:Automated} and classic
        automated mechanism design~\cite{Conitzer2002:Complexity}) are not
        possible.

    \item The Gini mechanism offers a price discovery process.  We believe any
        mechanism that asks for a valuation is not realistic for ICOs.  Before
        an ICO starts, the agents have difficulty in the valuation
        of the coin.  The value of a cryptocurrency depends heavily on its
        popularity.  The Gini mechanism only asks for the agents' budgets
        (\emph{i.e.}, how much do you plan to invest). Based on the budgets,
        the mechanism calculates a price.  The Gini mechanism has the property
        that it produces higher prices if there are higher budgets and more agents. That
        is, the more popular the coin is, the higher the price gets. During the running of the mechanism,
        the agents have the ability to adjust their budgets. For example, if the
        price is too high for an agent, then she could pull out (reduce her
        budget to $0$).  If the price gets too low, then the agent would max
        out her investment (honestly report her maximum budget). 
        We show
        that a pure strategy
        $(\epsilon_1,\epsilon_2,\ldots,\epsilon_n)$-equilibrium exists. Under
        this equilibrium, agent $i$ is playing a pure strategy (offering a
        single budget) and this is at most $\epsilon_i$ away from her best
        response in terms of utility. The $\epsilon_i$ depend on the actual
        data. In our experiments, the $\epsilon_i$ are always \emph{tiny
        compared to the agents' maximum budgets}.

    \item For most of the agents, 
        the strategies are straight-forward.
        They either max out their budgets or completely pull out of the
        investment.  Only a handful of agents have room for calculated
        equilibrium behaviors --- they face the situations where investing too
        much results in a price that is too high, while investing too little
        results in little utility gain. A carefully chosen budget can
        create a nice balance between the price increment and her
        marginal utility gain.

    \item The mechanism produces nearly optimal revenue in experiments. In
        terms of theoretical guarantee, \emph{asymptotically} (when the number
        of agents goes to infinity), under natural technical assumptions, the
        mechanism's revenue converges to the optimal revenue.

\end{itemize}

As a last but not least contribution of this paper, we plan to release our
collected dataset, which can be used to analyze user behaviors in ICOs and
for agent-based simulations of ICOs.  For convenience, in our datasets, all
monetary amounts are provided in both Ether and USD, based on real-time
exchange rate at the time of the transactions.  Our dataset consists of six
popular ICOs, collected according to the following selection rules:

\begin{description}
   \item[Rule 1:] More than 10 millions (USD) raised.
   \item[Rule 2:] Using the Dutch auction as the ICO model.
   \item[Rule 3:] The number of transactions is more than 500.
\end{description}





\section{Model Description} The unit size of digital currencies tend to be tiny.
For example, \emph{satoshi} is the name for the smallest unit of
bitcoin, which equals one hundred millionth of a bitcoin. In our model, we
treat coins as \textbf{infinitely divisible}, which allows us to normalize the
number of coins to $1$.  That is, we are selling \textbf{one divisible item
(one coin)} to $n$ agents.  For presentation purposes, we say that agent $i$
receives $a_i$ coin ($0\le a_i\le 1$) if she receives $a_i$ fraction of the
coin.  Agent $i$'s type is denoted as $(v_i, b_i)$, where $v_i$ is her
valuation for the coin\footnote{If the whole coin reserve is released via the
ICO mechanism, then an agent's valuation is essentially her view of the
\emph{market cap}.} and $b_i$ is her budget. For now, we defer any discussion
on the difference between an agent's private true budget limit and her reported
budget limit. We will discuss the agents' strategies in Section~\ref{sec:equilibrium}.

Since we are selling currencies, we believe all agents should face the same
\emph{exchange rate}. In an ICO mechanism, the agents would pay monetary
payments (\emph{e.g.}, USD or other cryptocurrencies like Ether) in exchange of
a fraction of the coin.  Let $i$ and $j$ be two agents, each receiving $a_i$
and $a_j$ coin, and each paying $c_i$ and $c_j$. We should have
$\frac{c_i}{a_i}=\frac{c_j}{a_j}$. This exchange ratio can also be interpreted
as the \emph{price} of the coin.  Essentially, we require that our
mechanism offers the same price $p$ to all agents. If we charge $c$ from an
agent, then this agent should receive $\frac{c}{p}$ coin.

For agent $i$, having a budget of $b_i$ does not necessarily mean that the
mechanism will charge her exactly $b_i$.  Let $c_i$ be agent $i$'s actual
spending under the mechanism. The budget constraint is $0\le c_i\le b_i$. Agent
$i$ receives $\frac{c_i}{p}$ coin and her utility equals \[
    \frac{c_i}{p}(v_i-p) = \frac{c_iv_i}{p}-c_i \]
The agents aim to maximize their utilities. Again, we defer any discussion
on strategies to Section~\ref{sec:equilibrium}.

A mechanism outputs the price $p$ and an allocation $(a_1,a_2,\ldots,a_n)$.
$a_i\ge 0$ for all $i$ and $\sum a_i=1$.  An allocation is feasible only if it
honours the budget constraint: $a_ip\le b_i$ for all $i$.

We also introduce a new mechanism design constraint to limit the
degree of inequality in our allocations. The popular \emph{Gini index}
is used to measure allocation inequality. We set a constant \emph{Gini cap} $g$ on the
Gini index. That is, a feasible allocation's Gini index should never exceed
$g$.  The Gini cap is a mechanism parameter chosen by the mechanism designer, with $0<g<1$.
A higher Gini cap implies that we have a higher tolerance on allocation
inequality.

The standard way to compute the Gini index is as follows.  Let
$(a_1,a_2\ldots,a_n)$ be the allocation.  We sort the $a_i$ in ascending order
to obtain the $y_i$. So $y_1$ is the smallest value among the $a_i$ and $y_n$
is the largest value among the $a_i$.  The Gini index equals
\[ \frac{2\sum_{i=1}^n iy_i}{n\sum_{i=1}^n y_i} - \frac{n+1}{n} =\frac{2\sum_{i=1}^n iy_i}{n} - \frac{n+1}{n} \]

However, there is one issue with the above definition.  The whole point of
considering the Gini index is to ensure allocation equality.  Generally speaking, we want
to avoid situations where some agents receive too little while some other
agents receive too much. If an agent has a huge budget, to prevent her from
receiving too much, the mechanism can simply set an \emph{investment cap}.
That is, any investment beyond the cap is not accepted. On the other hand, if an
agent has a tiny budget, to prevent her from receiving too little, we have to
reduce the price to accommodate her tiny budget, which may significantly
hurt the mechanism performance --- after all, an ICO mechanism's goal is to
raise money.

Let us consider an extreme example with $g=0.5$ and $n=100$. Let us assume that
$75$ agents have $0$ budgets and the remaining $25$ agents have very high
budgets.  In this case, there are no feasible allocations. All allocations'
Gini indices exceed $0.5$.  To show this, we note that
$y_1=y_2=\ldots,y_{75}=0$, so the Gini coefficient becomes \[
    \frac{2\sum_{i=76}^{100} iy_i}{100\sum_{i=76}^{100} y_i} - \frac{101}{100}
    \ge \frac{2\cdot 76}{100} - 1.01 = 0.51 \]

Actually, for any constant $g<1$, we can find type profiles that make it
impossible to allocate (to meet the Gini cap).  We do not wish to simply fail
the ICO in these scenarios.  Instead, we allow the mechanism to \emph{ignore
agents who receive nothing from the Gini index calculation}.  For the above
example, if we ignore all $75$ agents who receive nothing, then feasible
allocation is possible. We focus on the remaining $25$ agents, who all have
positive budgets.  We could simply allocate every agent an equal share
($a_i=\frac{1}{25}$), by setting the price low enough so that every agent can
afford $\frac{1}{25}$ coin.  Equal sharing has a Gini index of $0$.  Therefore, ignoring
agents saves us from infeasible situations. There are
two arguments for the flexibility of ignoring agents who receive nothing:

\begin{itemize} \item We are not considering all $7$ billion people when
            calculating the Gini index anyway.  People who have not joined the
            ICO are not fundamentally different from agents who receive nothing
            in the ICO.

    \item We are able to achieve much higher revenue under this assumption. We
        can construct example situations where \emph{one tiny-budget agent}
        becomes the revenue bottleneck. By not allocating anything to this
        agent and not including her in the Gini index calculation, we sometimes
        can increase the revenue by infinite many times.  \end{itemize}

Formally speaking, for our specific model, we calculate the Gini index as follows:

\begin{definition}[Flexible Gini Index]
We allow the mechanism to pick the number of winners $k$ from a set $K$.
$K$ includes all the allowed winner numbers.

The $n-k$ non-winners all receive nothing, and they are not included in the
    Gini index calculation.

Let $(a_1,a_2\ldots,a_k)$ be the allocation for the $k$ winners.
    We still have that $\sum_{i=1}^k a_i=1$ and $a_i\ge 0$.
We sort the $a_i$ in ascending order to obtain the $y_i$. So $y_1$ is the smallest value
among the $a_i$ and $y_k$ is the largest value among the $a_i$.
The Gini index for $k$ winners is defined as:
\begin{equation}\label{eq:ginik}
\frac{2\sum_{i=1}^k iy_i}{k\sum_{i=1}^k y_i} - \frac{k+1}{k}
\end{equation}
\end{definition}

Here are a few example setups for $K$:

\begin{itemize}
    \item $K=\{n\}$: No agents can be ignored. We fall back to the standard definition.
    \item $K=\{\ceil{0.5n}, \ceil{0.6n}, \ceil{0.7n}, \ceil{0.8n}, \ceil{0.9n}, n\}$:
The mechanism can ignore $50\%, 40\%, \ldots, 0\%$ of the agents.
    \item $K=\{\ceil{0.5n}, \ceil{0.5n}+1,\ldots, n\}$: The mechanism picks at least a half of the agents as winners.
\end{itemize}

$K$ is also the mechanism's parameter.  $min(K)$ is the minimum number of
winners. We should not allow the mechanism to pick too few winners. Having only
one winner leads to a very nice Gini index --- it is always $0$, but it is also
meaningless.

The set of feasible allocations is determined by the allowed winner numbers
$K$, the Gini cap $g$, the price $p$, and finally the agents' budgets.  If we
increase the price $p$, then every agent's allocation upper limit is reduced.
Therefore, the set of feasible allocations either stays the same or shrinks.
The total revenue of a mechanism is exactly the price (because we have only one
coin for sale). Therefore, to maximize revenue, a natural idea is to push up
the price to the point so that any further price increment makes feasible
allocation impossible.  We propose the Gini mechanism based on
exactly this idea.

\vspace{.05in} \begin{center} \noindent\fbox{\parbox{\textwidth}{
        \textbf{Informal description of the Gini mechanism:} The Gini mechanism
        does not ask for the agents' valuations at all.  The mechanism produces
        a price based on the agents' budgets alone.

We start with an infinitesimally small price and raise the price until any
    further price increment results in no feasible allocations.  At the final
    price, the feasible allocation is unique subject to tie-breaking.
}}
\end{center}

\section{Formal Mechanism Description} We start with a procedure that will be
used as a building block of the Gini mechanism.  The procedure answers the following
question: given the agents' budgets (the $b_i$), given the price $p$ and the
number of winners $k$, what is the allocation that minimizes the Gini index?

In the context of the above question, we do not have a Gini cap. Instead, we
search for an allocation that minimizes the Gini index. An allocation
$(a_1,a_2,\ldots,a_n)$ is feasible if it satisfies the following:

\begin{itemize}
    \item $a_i\ge 0$ and $\sum a_i=1$.
    \item $a_ip\le b_i$. $a_ip$ is agent $i$'s spending and $b_i$ is the budget limit.
        If $a_ip=b_i$, then we say this agent is \emph{maxed out}.
    \item At least $n-k$ elements of the $a_i$ are $0$s.
\end{itemize}

Without loss of generality, we assume $0\le b_1\le b_2\le\ldots\le b_n$.  Let
$(a_1^*, a_2^*, \ldots,a_n^*)$ be the feasible allocation that minimizes the
Gini index. We first notice that it is without loss of generality to assume
$0\le a_1^*\le a_2^*\le \ldots\le a_n^*$.  The reason is that if we have $i<j$
and $a_i^*>a_j^*$, then swapping $a_i^*$ and $a_j^*$ results in a feasible
allocation with the same Gini index.  This also means that we can set $a_1$ to
$a_{n-k}$ to $0$s as we have only $k$ winners.

We then consider only the winners (agent $n-k+1$ to $n$).  If agent $i$ is
not maxed out ($a_i^*p < b_i$), then agent $i+1$ must receive the
same allocation amount (\emph{i.e.}, $a_{i+1}^*=a_i^*$).  For if it is
not, we could increase $a_i^*$ by a small $\epsilon$ (still affordable by $i$)
and decrease $a_{i+1}^*$ by the same $\epsilon$ (ensuring that we still have
$a_i^*\le a_{i+1}^*$). We end up with a feasible allocation with a strictly
smaller Gini index.

Furthermore, among the $k$ winners, let $t$ be the agent that is not maxed out
with the smallest index $t$.  Agent $t+1$ must have the same allocation amount
as $t$. This implies that $t+1$ is also not maxed out, which implies that agent
$t+2$ should also have the same allocation. That is, the only allocation
structure we need to consider is

\[(\underbrace{0,0,\ldots,0}_{n-k},\frac{b_{n-k+1}}{p},\frac{b_{n-k+2}}{p},\ldots,\frac{b_{t-1}}{p},\underbrace{C,C,\ldots,C}_{n-t+1})
\]

In this allocation, we refer to $C$ as the \emph{allocation cap}. $n-t+1$ is the number of
capped agents.

By definition of $C$, we have $\frac{b_{t-1}}{p}< C\le \frac{b_{t}}{p}$, so the
above structure can be rewritten into
\begin{equation}\label{eq:minginiform}
(\underbrace{0,0,\ldots,0}_{n-k},\min\{\frac{b_{n-k+1}}{p},C\},\min\{\frac{b_{n-k+2}}{p},C\},\ldots,\min\{\frac{b_n}{p},C\})
\end{equation}

The total allocation must be $1$, hence $C$ must
satisfy the following equation:
\begin{equation}\label{eq:mingini}
\sum_{i=n-k+1}^n\min\{\frac{b_{i}}{p},C\}=1
\end{equation}

The only constraint on $C$ is that it serves as a cap, so $0< C\le \frac{b_n}{p}$.
The left side of Equation~\eqref{eq:mingini} is strictly increasing in $C$.
The only way Equation~\eqref{eq:mingini} does not have a solution is when
$C$ has already reached $\frac{b_n}{p}$ but the left side is still less than $1$.
\begin{equation}\label{eq:mingini1}
\sum_{i=n-k+1}^n\frac{b_{i}}{p}<1
\end{equation}

If this happens, then we do not have any feasible allocations. For convenience,
we define the minimum Gini index to be $1$ for this case.  When
Equation~\eqref{eq:mingini} has solutions, the solution is unique. By solving
for $C$, we can find the allocation that minimizes the Gini index, based on
Expression~\eqref{eq:minginiform}.  The optimal (Gini-index-minimizing)
allocation is unique subject to a consistent tie-breaking rule.  The allocation
essentially does not allocate to the lowest $n-k$ agents (in terms of budgets).
This is the only place where we need tie-breaking. We may simply break ties by
favoring agent $i$ over agent $j$ if $i>j$.  The optimal allocation then sets
an allocation cap $C$. All agents below $C$ max out their budgets and all
agents at least $C$ can only get $C$.



Let us then consider the relationship between $p$ and the minimum Gini index,
while fixing the agents' budgets and the winner number $k$.  We define $g^*(p)$
to be the minimum Gini index for price $p$.

\begin{proposition}\label{prop:lower}
    When $0<p\le kb_{n-k+1}$,
    $g^*(p)=0$.
\end{proposition}

\begin{proof} When $0<p\le kb_{n-k+1}$, we set $C=\frac{1}{k}$ and
Equation~\eqref{eq:mingini} is satisfied and every winner receives the same
allocation $\frac{1}{k}$, which corresponds to a Gini index of $0$.
\end{proof}

The Gini mechanism starts with an infinitesimally small price.
Proposition~\ref{prop:lower} basically says that initially feasible allocations
must exist if at least $k$ agents have positive budgets, regardless of the Gini cap.\footnote{Even if less than $k$ agents have positive budgets, we may still be able to find an infinitesimally small price that makes feasible allocation possible. If no feasible allocation exists even with infinitesimally small price, then our mechanism fails. In this case, the mechanism designer should consider raising the Gini cap.}  The Gini mechanism then increases the
price until any further increment results in no feasible allocations.  The next
proposition guarantees that when the price is large enough, no feasible
allocations exist.

\begin{proposition}\label{prop:upper}
    When $p> \sum_{i=n-k+1}^nb_i$, $g^*(p)=1$.
\end{proposition}

\begin{proof}
    This is based on Inequality~\eqref{eq:mingini1}.
\end{proof}

Proposition~\ref{prop:lower} and~\ref{prop:upper} together are not enough.  How
can we determine when we have reached an optimal price $p^*$ so that any
further increment results in no feasible allocations? Also, we need to prove
that $p^*$ exists. For example, it could be that the price can increase in
$[0,2)$ without causing infeasibility, but when the price reaches exactly $2$,
all of a sudden no feasible allocations exist. In this case, technically
speaking $p^*$ does not exist.  It cannot be $2$. Of course, in practice,
$2-\epsilon$ is fine. We will show that $p^*$ always exists. Situations like
the above do not occur.

\begin{proposition}\label{prop:cont}
    As we increase $p$ in $(0,\sum_{i=n-k+1}^nb_i]$, $g^*(p)$ is
    continuously nondecreasing in $p$.
\end{proposition}

\begin{proof} As we increase $p$, every agent's allocation upper limit
    ($b_i/p$) is decreased.  So the set of allowed allocations either stays the
    same or shrinks.  Therefore, the minimum Gini index is nondecreasing with
    the price.

    Based on Equation~\eqref{eq:mingini}, the cap $C$ is continuous in $p$.
The allocation with minimum Gini index (Expression~\eqref{eq:minginiform})
changes continuously in $p$. Therefore, the Gini index changes continuously in
$p$.  \end{proof}

Combining all propositions, we have the following theorem:

\begin{theorem}\label{thm:maxmiumprice}Fixing the agents' budgets and the number of winners $k$, we
    define the maximum price $p^*$ to be the maximum price where the minimum
    Gini index is at most the Gini cap $g$. At price $p^*$, there are feasible
    allocations. Any increment in $p^*$ results in no feasible allocations.

$p^*$ exists and the corresponding feasible allocation is unique subject to tie-breaking.
\end{theorem}

\begin{proof} Based on Proposition~\ref{prop:upper}, if the minimum Gini index
    at $p=\sum_{i=n-k+1}^nb_i$ is at most $g$, then this is the maximum price.
    Any higher price results in a Gini index of $1$.  At this price, the only
    feasible allocation is that the lowest $n-k$ agents in terms of budgets do
    not receive anything and the highest $k$ agents all max out.

    We then consider situations where the minimum Gini index at
    $p=\sum_{i=n-k+1}^nb_i$ is strictly higher than $g$.  When $p=kb_{n-k+1}$,
    the minimum Gini index is $0$, so $p^*\in [kb_{n-k+1},
    \sum_{i=n-k+1}^nb_i]$.  Let $S=\{p|g^*(p)=g, p \in [kb_{n-k+1},
    \sum_{i=n-k+1}^nb_i]\}$.  Given that $g^*(p)$ is continuously nondecreasing
    based on Proposition~\ref{prop:cont}, $S$ contains either one point, or $S$
    is a closed interval.  In both cases, $p^*=\max S$ exists.  At price $p^*$,
    the only feasible allocations are the Gini-index-minimizing allocations,
    which are unique subject to tie-breaking. If an allocation is not Gini
    minimizing and is feasible, then the minimum Gini index must be strictly
below $g$, which means that the price can still be pushed up due to the
continuity between the price and the minimum Gini index.  \end{proof}

The maximum price can be calculated via a simple binary search
inside the interval $[kb_{n-k+1},\sum_{i=n-k+1}^nb_i]$.
Next, we formally define the Gini mechanism.

\begin{definition}[Gini mechanism] $g$ and $K$ are mechanism parameters.
    Given the agents' budgets, for every allowed winner number $k$ in $K$, we find the maximum
    price $p^*_k$ where the minimum Gini index is at most the Gini cap $g$.
    We then choose the price to be the overall maximum $\max_{k\in K}p^*_k$.
    If there are multiple $k$ values with the same overall maximum price, then we break ties
    by favoring more winning agents. We pick the unique feasible allocation corresponding
    to the price.
\end{definition}

\section{Equilibrium under the Gini Mechanism}\label{sec:equilibrium} In this
section, we start the discussion on the agents' strategies.  The way we would
implement the Gini mechanism in practise is as follows: We announce a time
frame for the ICO. During the time frame, the agents can join/leave anytime,
and can change their investment amounts (budgets) anytime. (Technically,
joining and leaving are special cases of changing budgets.) The Gini mechanism
maintains the current price and allocation throughout the time frame.  We
assume that the time frame is long enough so that at some point, after all
interested agents have joined, an equilibrium on the budgets are to be
reached.\footnote{In practise, most ICOs are conducted over the Ethereum
network. Here, an equilibrium will always be reached due to the
\emph{transaction fees}. In this paper's model, we do not consider transaction
fees.} The equilibrium price/allocation eventually become the final
price/allocation.

Our discussion involves three closely related concepts:

\begin{itemize}

    \item We use $b_i^M$ to denote agent $i$'s true maximum budget. This is
        $i$'s private information.

    \item We use $b_i$ to denote agent $i$'s reported budget.  Agent $i$ can
        report arbitrary nonnegative budget, including reporting a value above
        $b_i^M$.

    \item We use $c_i$ to denote agent $i$'s spending.
        The spending is always at most the reported
        budget. That is, $0\le c_i\le b_i$.

\end{itemize}

We use $\vec{b}=(b_1,b_2,\ldots,b_n)$ to denote the current budget profile. We
can also call this the agents' current strategies. We still assume that $0\le b_1\le b_2\le \ldots\le b_n$.
For every winner number $k\in K$, we use
$p_k^*(\vec{b})$ to denote the maximum price for $k$ winners and budget profile
$\vec{b}$.  The Gini mechanism's price $p^*(\vec{b})$ for budget profile
$\vec{b}$ is the maximum over the $p_k^*(\vec{b})$ for $k\in K$.


Now we introduce a few propositions, which will be used for analyzing agents'
strategies.

\begin{proposition}\label{prop:bpinc}
    For any $k$, $p_k^*(\vec{b})$ is nondecreasing in every coordinate (every $b_i$).
    This also means that $p^*(\vec{b})$ is nondecreasing in every coordinate, because
    the maximum of monotone functions are still monotone.
\end{proposition}

\begin{proof} When we increase agent $i$'s budget $b_i$, feasible allocations
remain feasible allocations. If previously we can push the price to a certain point,
then we still can (and may be able to push the price up even more).
\end{proof}

\begin{proposition}\label{prop:bpcont}
    For any $k$, $p_k^*(\vec{b})$ is continuous in every coordinate (every $b_i$).
    This also means that $p^*(\vec{b})$ is continuous in every coordinate, because
    the maximum of continuous functions are still continuous.
\end{proposition}

\begin{proof} We focus on a specific $k$.  If we change $b_i$ to $\alpha b_i$
    for \emph{every} $i$, then the maximum price also changes by $\alpha$
    times (to offset the change in budgets).  Let $p$ be the maximum price. For an arbitrary $\epsilon>0$, to
    change the price from $p$ to $p+\epsilon$, we can increase every $b_i$ to
    $b_i\frac{p+\epsilon}{p}$.  This also means that for a specific $i$, if the
    increment in $b_i$ is at most $b_i\frac{p+\epsilon}{p}-b_i$, and we do not
    change the other coordinates (this means less increment in price because
the price is monotone in every coordinate), then the price increment is at most
$\epsilon$. The same argument works for decrement. \end{proof}

\begin{proposition}\label{prop:bminbmax}
    Agent $i$ faces a minimum investment amount $b_{min}(\vec{b}_{-i})$
    and a maximum investment cap $b_{max}(\vec{b}_{-i})$. Both are determined
    by the other agents' budgets (\emph{i.e.}, $\vec{b}_{-i}$).
    Agent $i$'s spending is $0$ when her budget is below $b_{min}(\vec{b}_{-i})$
    and her spending stays at $b_{max}(\vec{b}_{-i})$ when her budget grows beyond the cap.
    In between the minimum and the maximum investment amounts,
    the price $p^*$ strictly increases in $i$'s reported budget.
\end{proposition}

\begin{proof} We raise $i$'s budget from $0$. At some point when her budget
    reaches $b$, for the first time she becomes a winner under the Gini
    mechanism.  Due to tie-breaking, it could be that $i$'s budget must be
    strictly above $b$ for her to become a winner. We ignore this technicality.
    Once $i$ becomes a winner. That means at the current price, $i$ is a winner
    in at least one feasible allocation. Let the set of feasible allocations
    where $i$ is a winner be $S_i$. If $i$ is not capped in at least one
    allocation in $S_i$, then by increasing $i$'s budget, the cap of this
    allocation decreases at the current price, which leads to strictly smaller
    Gini index. This then means the overall price $p^*$ should strictly increase as a result.
    If $i$ is capped in all feasible allocations in $S_i$, then any increment
in $i$'s budget has no effect anywhere, this means that $i$ has already reached
her maximum investment cap.  \end{proof}

\begin{proposition}\label{prop:bpdivinc}
    For any $k$ and $i$, $b_i/p_k^*(\vec{b})$ is nondecreasing in $b_i$.
    This also means that $b_i/p^*(\vec{b})$ is nondecreasing in $b_i$, because
    the minimum of monotone functions are still monotone.
\end{proposition}

\begin{proof} If $b_i$ is small and $i$ is not a winner, then increasing $b_i$
    has no effect on the price.  If $b_i$ is large and $i$ is capped, then
    again, the increment in $b_i$ has no effect on the price. When price stays
    the same, $b_i/p_k^*(\vec{b})$ is nondecreasing in $b_i$.  When $b_i$ is a
    winner and not capped, when increasing $b_i$, the price gets strictly
    higher. If $i$'s allocation $b_i/p_k^*(\vec{b})$ decreases, then all
    uncapped agents' allocations decrease. If any agent is capped, then the cap
    increases.  This results in strictly higher Gini index, which makes the
allocation infeasible. If no agents is capped, then
$b_i/p_k^*(\vec{b})=b_i/\sum_{j=n-k+1}^nb_j$, which is also nondecreasing in $b_i$.
\end{proof}

\begin{proposition}\label{prop:bpderivative}
    Assuming $0<b_1<b_2<\ldots<b_n$,
   for any $i$,
    \[ \frac{\partial p^*_k(\vec{b})}{\partial b_i}\le\frac{g+3}{1-g}
    \]
    This also means that the partial derivative of $p^*(\vec{b})$ against $b_i$ is
    bounded above by $\frac{g+3}{1-g}$.
\end{proposition}

We focus on budget profiles without ties for convenience. In an actual
equilibrium, if the budget profile contains ties, then we can simply perturb
the budgets infinitesimally to remove ties. We have shown that the price
function is continuous, so perturbing the budgets will only change the agents'
utilities infinitesimally.

\begin{proof} We focus on a specific winner number $k$ and a specific budget
    profile $(b_1,b_2,\ldots,b_n)$ with $0<b_1<b_2<\ldots<b_n$.

    Let $z=\floor{(n-k)+\frac{k}{2}(g+\frac{k+1}{k})}$.  Let us analyze the
    derivative with respect to $b_i$ when $i>z$.  If $b_i$ is already capped
    under this budget profile, then the derivative is $0$.  So we only need to
    consider the situation where $b_i$ is not capped. We reduce $b_i$ to
    $b_i-\epsilon$ and change the price from $p^*$ to $p^*-\epsilon$.  We
    consider a new allocation where every other agent's spending stays the
    same, but $i$'s spending is reduced by $\epsilon$. This is still a feasible
    allocation but may not meet the Gini cap.  We consider the Gini index of
    the new allocation. Let us consider the definition of the Gini index
    in~\eqref{eq:ginik}. For the new allocation, the only change is a reduction
    in proportion of $y_i$ in both the numerator and the denominator. The
    numerator is multiplied by $2(i-n+k)$ (the indices are from $1$ to $k$ for
    the $k$ winners in~\eqref{eq:ginik}). The denominator is multiplied by $k$.
    Since $i>z$, we have $i\ge z+1$ and \[\frac{2(i-n+k)}{k}\ge
    \frac{2(z+1-n+k)}{k}\ge g+\frac{k+1}{k} \] We want the first term
    of~\eqref{eq:ginik} to be at most $g+\frac{k+1}{k}$ to meet the Gini cap.
    So lowering $y_i$'s proportion only helps this goal. This means that the
    new allocation also meets the Gini cap. So the overall price should be at
    least $p-\epsilon$. In conclusion, the derivative is at most $1$ when
    $i>z$.

    We then consider $i\le z$. If $i\le n-k$, then $i$ is not a winner and the
    derivative is $0$. We only present the proof for $i=n-k+1$. This is the
    winner with the lowest budget. This budget has the highest impact on the
    price. $b_{n-k+1}$ is not capped, for otherwise the Gini index equals $0$.
    Therefore, we can adjust $b_{n-k+1}$ both up and down. $n-k+1$'s
    spending is the same as her budget.

    We first consider the case where $b_{z+1}$ is not capped.
    We define a few terms (the $c_i$ are the actual spendings of the agents):
    \begin{itemize}
        \item $x'=\sum_{i=n-k+1}^z(i-n+k)c_i$
        \item $x=\sum_{i=n-k+1}^zc_i$
        \item $y'=\sum_{i=z+1}^n(i-n+k)c_i$
        \item $y=\sum_{i=z+1}^nc_i$
    \end{itemize}
    The first term of~\eqref{eq:ginik} is then $r=\frac{x'+y'}{x+y}$.
    When the Gini cap is met, $r=\frac{k}{2}(g+\frac{k+1}{k})$.
    (If the Gini cap is not met, then all $k$ agents spend all their
    budgets, in which case any change in budget corresponds to a derivative of at most $1$.)
    We use $s$ to denote $x+y$.

    We change $b_{n-k+1}$ to $b_{n-k+1}+\epsilon$.  We consider a new
    allocation with the following spendings. Agent $n-k+1$'s spending is
    increased from $b_{n-k+1}$ to $b_{n-k+1}+\epsilon$. Other agents from $n-k+2$
    to $z$ keep their current spendings. The agents from $z+1$ to $n$ reduce
    their spendings by a factor of $\alpha$.

    The first term of the Gini index of this new allocation is then
    $\frac{x'+\epsilon+\alpha y'}{x+\epsilon+\alpha y}$. We set $\alpha$ so
    that this term equals $r$.
    \[\alpha=\frac{(x-x')\epsilon+(y-y')\epsilon+x'y-xy'}{x'y-xy'}<1\] This
    implies that the new allocation is still feasible given the above $\alpha$.
    Since we assume agent $z+1$ is not capped, there is room for pushing down
    the spendings of agent $z+1$ to $n$ by multiplying the original spendings by $\alpha$.
    The price corresponding to the new allocation is the denominator of the
    first term of the Gini index, which is $x+\epsilon+\alpha y$ (under the Gini mechanism, the price is always equal to the total spending).
    The derivative
    of $x+\epsilon+\alpha y$ against $\epsilon$ equals
    \[\frac{(x+y)(y'-y)}{xy'-x'y}\]

    We simply the derivative:
    \[\frac{(x+y)(y'-y)}{xy'-x'y}=\frac{xy'+yy'-yy-xy+(x'y-x'y)}{xy'-x'y}\]
    \[=1+y\frac{y'-y-x+x'}{xy'-x'y}=1+y\frac{rs-s}{xy'-x'y}\]
    \[=1+y\frac{rs-s}{xy'-x'y+(x'x-x'x)}=1+y\frac{rs-s}{xrs-x's}\]
    \[=1+y\frac{r-1}{xr-x'}=1+y\frac{r-1}{xr-x'+ry-ry}=1+y\frac{r-1}{y'-ry}\]
    To maximize the derivative, we minimize $\frac{y'-ry}{y}=\frac{y'}{y}-r$ instead.
    $\frac{y'}{y}$ is minimized when all $b_i$ are the same for $i\ge z+1$.
    The ratio is minimized to
    \[\frac{(z+1-n+k+n)(n-z)}{2(n-z)}=\frac{z+1-n+k+n}{2}\ge\frac{r+n}{2}\]

    So $\frac{y'-ry}{y}$ is minimized to $\frac{n-r}{2}$.
    The derivative is maximized to
    \[1+\frac{2(r-1)}{n-r}=
    1+\frac{2k(g+\frac{k+1}{k})-4}{2n-k(g+\frac{k+1}{k})}
    =1+\frac{2kg+2(k+1)-4}{2n-kg-(k+1)}
    \]
    The above increases with $k$, but $k$ is at most $n$, so the above is maximized to
    \[1+\frac{2ng+2(n+1)-4}{2n-ng-(n+1)}=
    1+\frac{(2g+2)n-2}{(1-g)n-1}=
    \frac{(g+3)n-3}{(1-g)n-1}
    \]
    The above expression increases with $n$. When $n$ goes to infinity, we
    have the final upper bound $\frac{g+3}{1-g}$.

    When $z+1$ is capped, we consider changing $b_{n-k+1}$ to $b_{n-k+1}-\epsilon$
    instead. All analysis is almost identical. After the change, $\alpha>1$, but
    if $z+1$ and future agents are capped, there is room for increasing the spendings.
    \end{proof}

Given the above propositions, we demonstrate that for most agents, the
strategies are fairly simple.  We start with a sufficient condition for an
agent to report $0$ budget.

\begin{proposition}\label{prop:pullout}
    If $p^*(0, b_{-i})\ge v_i$, then an agent's best strategy is to report a budget of $0$.
\end{proposition}

In experiments, generally we have that $p^*(0,b_{-i})$ is very close to the
$p^*(b_i,b_{-i})$. Our observation is that an agent's impact on the overall
price is very limited. This proposition essentially says that if an agent's
value is below (about) the current price, then she wants to leave (by setting the
budget to $0$).

\begin{proof}
    The price is nondecreasing in the agents' budgets. $p^*(0,b_{-i})$ is the lowest
    price $i$ faces by unilateral budget change. If this price is still at least
    her valuation, then she does not want to invest.
\end{proof}

Now we provide a sufficient condition for an agent to report her true maximum
budget. It is \emph{not as simple as}, for example, if $p^*(b^M_i,b_{-i})\le
v_i$, then an agent would report her actual maximum. For reporting the maximum
budget, an agent's valuation must be \emph{slightly} higher than the
$p^*(b^M_i,b_{-i})$. This is to ensure that the utility gain for buying more is
always greater than the utility loss caused by price increment for the existing
purchase.

\begin{proposition}\label{prop:maxout} If $v_i\ge
    \frac{p^*(H_i,b_{-i})}{1-a_i^M\frac{g+3}{1-g}}$, then agent $i$'s best
    strategy is to report her true maximum budget.

    Here, $H_i=\min\{b^M_i, b_{max}(\vec{b}_{-i})\}$, which represents the
highest ``effective'' budget for agent $i$. Any higher budget is essentially the same or
violates the budget constraint.  $a^M_i$ is $i$'s allocation when her budget is $H_i$.
\end{proposition}

In experiments, generally we have that $p^*(H_i,b_{-i})$ is very close to
$p^*(b_i,b_{-i})$.  $a_i^M$ is generally very small for most agents. For
example, if there are $3000$ agents in total, then we know for sure that at
most $1000$ agents can get allocations at least $0.1\%$. That means for the
remaining $2000$ agents, we have $a_i^M\le 0.001$. Let us consider $g=0.6$, so
$\frac{g+3}{1-g}=9$. For an agent among these $2000$ agents, if
her valuation is at least (about) the current price, divided by $0.991$, then
she wants to report her true maximum budget.

\begin{proof}
Due to Proposition~\ref{prop:bminbmax}, agent $i$ faces a minimum and a maximum
investment amount. If her true maximum budget is below the minimum amount, then
this agent is irrelevant. We only consider agents who can meet the minimum
investment amounts.  We use $u_i(b_i,b_{-i})$ to denote $i$'s utility. Let us
analyze the derivative of $u_i$ against $b_i$ when $b_i$ is identical to the
spending $c_i$. That is, we consider $b_i$ in between the minimum investment
amount and $H_i$.
\[
    u_i(b_i,b_{-i}) = \frac{b_i}{p^*(b_i,b_{-i})}(v_i-p^*(b_i,b_{-i})) = \frac{v_ib_i}{p^*(b_i,b_{-i})}-b_i
\]
\[
    \frac{\partial u_i(b_i,b_{-i})}{\partial b_i} = \frac{p^*(b_i,b_{-i})v_i-\frac{\partial p^*(b_i,b_{-i})}{\partial
    b_i}v_ib_i}{p^*(b_i,b_{-i})^2}-1
\]
\[
    =\frac{v_i}{p^*(b_i,b_{-i})}(1-\frac{b_i}{p^*(b_i,b_{-i})}\frac{\partial p^*(b_i,b_{-i})}{\partial
    b_i})-1
\]
\begin{equation}\label{eq:maxout1}
\ge \frac{v_i}{p^*(b_i, b_{-i})}(1-a_i^M\frac{g+3}{1-g})-1
\end{equation}
\begin{equation}\label{eq:maxout2}
\ge \frac{v_i}{p^*(H_i,b_{-i})}(1-a_i^M\frac{g+3}{1-g})-1
\end{equation}

In the last two steps, we relied on two facts. One is that due to
Proposition~\ref{prop:bpdivinc}, $a_i^M$ is the highest allocation $i$ can get.
The other is that the derivative is bounded according to Proposition~\ref{prop:bpderivative}.  \end{proof}


We compare the Gini mechanism's revenue against the \emph{first-best optimal
revenue}.  The first-best optimal revenue is calculated as an optimization
problem, assuming that we know all the agents' private valuations and private
budgets. Given a price $p$, we filter out all agents who can afford $p$, and
then derive the Gini-index-minimizing allocation based on the true maximum
budgets. If the Gini index is at most the Gini cap, then $p$ is an achievable
price.  We solve for the highest $p$.

\begin{theorem} Under the following assumptions, the Gini
    mechanism's equilibrium revenue approaches the first-best
    optimal revenue with probability $1$,  as the number of agents goes to
    infinity.
    \begin{itemize}
        \item The agents' private valuations are drawn \emph{i.i.d.} from a distribution with
            a value upper bound $U$. For any $\epsilon>0$, $Prob(v\in [U-\epsilon, U])>0$.
            This basically assumes that the upper bound is a meaningful upper bound, in the sense
            that there is a positive probability to draw a value close to the upper bound.
        \item The agents' private budgets are drawn \emph{i.i.d.} from a distribution with
            a positive expectation.
        \item There is a constant minimum winner number. Above that, all winner numbers are allowed.
    \end{itemize}
\end{theorem}

\begin{proof} We find a constant integer $n_0$ so that $Prob(b \ge U/n_0)>0$ and $n_0$ is
    at least the minimum winner number.
    $n_0$ can always be found.  There is a positive probability to draw an
    agent with valuation above $U-\epsilon$ and budget above $U/n_0$. As the
    number of agents goes to infinity, the probability of drawing $n_0+n_1$
    such agents equals $1$, where $n_1$ is another constant integer. Among
    these $n_0+n_1$ agents, for at least $n_0$ of them, the allocation is at
    most $\frac{1}{n_1}$. (For example, there are at most $1000$ agents who
    allocation is at least $\frac{1}{1000}$.) If the equilibrium price is less
    than $\frac{U-\epsilon}{1-\frac{1}{n_1}\frac{g+3}{1-g}}$, then according to
    Expression~\eqref{eq:maxout1}, these $n_0$ agents would max out their
    budgets under the equilibrium. Each of these $n_0$ agents has a
    budget that is at least $U/n_0$.  So at the equilibrium price, only
    allocating to these agents alone results in an allocation with a Gini index
    of $0$, which means that any equilibrium price strictly lower than
    $\frac{U-\epsilon}{1-\frac{1}{n_1}\frac{g+3}{1-g}}$ is not possible. (If it is lower, then
    it should be raised due to the existence of an allocation with $0$ Gini index.) Given
    that $\epsilon$ and $n_1$ are arbitrary constants, we have that the
equilibrium price can be made arbitrarily close to $U$. $U$ is obviously an upper
bound on the first-best optimal revenue.  \end{proof}

\section{Experiments} In this section, we evaluate the Gini mechanism using
real ICO data.  The selection criteria for our dataset are mentioned in
Section~\ref{sec:intro}. Our dataset is compiled based on previous Dutch
auction based ICOs.  Each ICO dataset contains the agents' budgets and their
entering prices to the Dutch auction. In an ICO Dutch auction, even if an
agent enters early, she still pays according to the ending price. So for agents
with small budgets (who cannot buy all the coins and stop the auction),
entering the auction when the price meets the valuation is a reasonable
strategy. For this reason, we use the entering prices as the agents'
valuations. In our experiments, we also included a list of generated bids. The
reason for this is that there are many observing agents whose data are missing
from our datasets.  An observing agent is someone who has low valuation. The
auction ended before such an agent logs her bid.  If the dataset contains $n$
agents, then we add in another $n$ generated agents.  The generated agents'
budgets are sampled from real budgets, and the valuations are just the ending
price times a random value from $0$ to $1$ (drawn according to the uniform
distribution).

We numerically compute an approximate Nash equilibrium. We first calculate the
first-best optimal price and use it to initialize the budget profile. An agent
reports the true maximum budget if her valuation is at least the first-best
optimal price, and reports $0$ otherwise. From this point on, we go through the
agents one-by-one and have each agent update the budget to her best response.
We stop when an equilibrium has been reached.

Our experiments involve thousands of agents, which seems very daunting when it
comes to equilibrium computation.  Fortunately, with the help of
Proposition~\ref{prop:pullout} and Proposition~\ref{prop:maxout}, we test these
two sufficient conditions and find that most agents either report the maximum
budget or $0$.

\begin{figure}
\centering
\includegraphics[width=\textwidth]{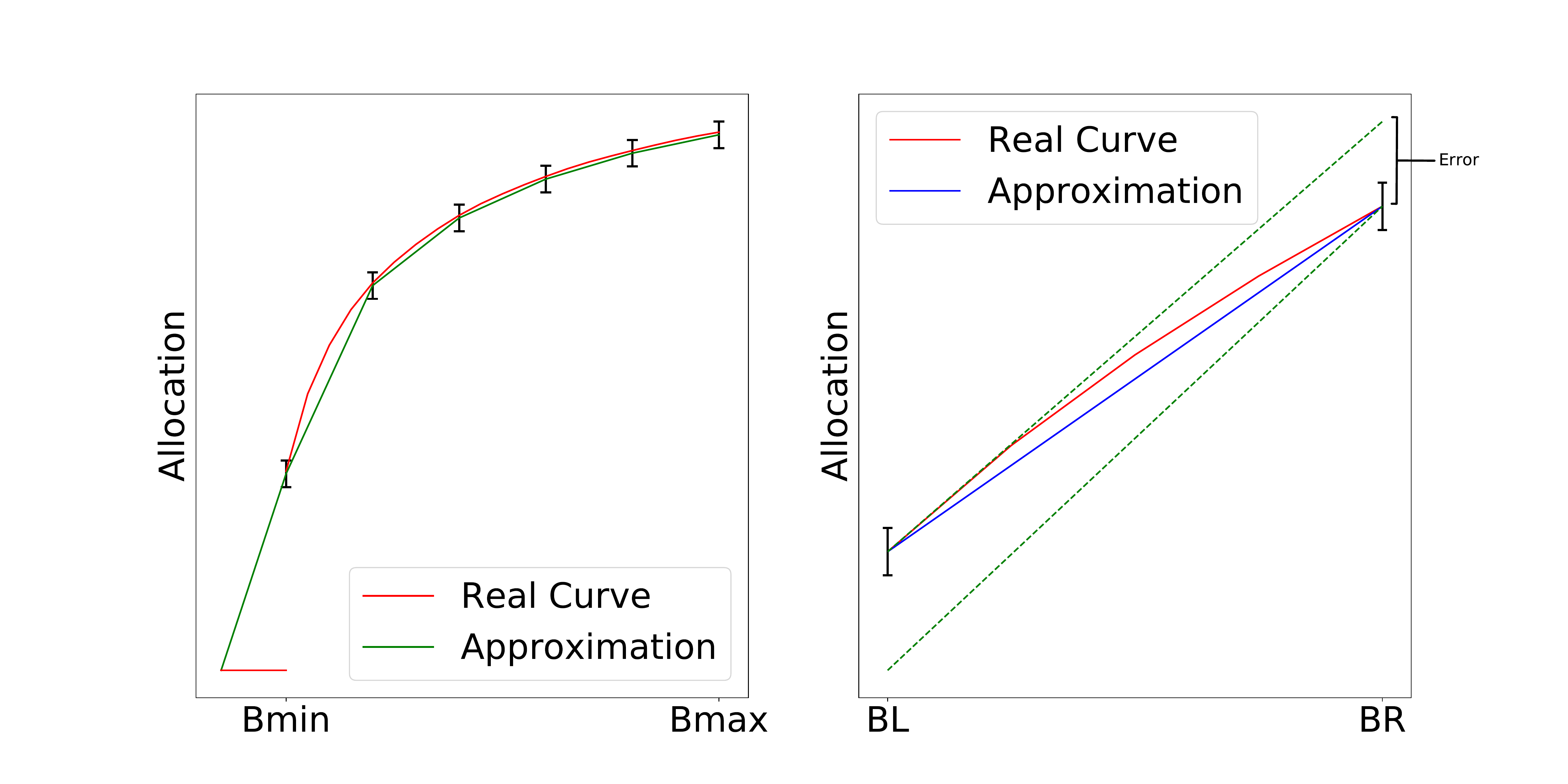}
\caption{Approximation for the Allocation Curve}\label{fig:only}
\end{figure}

For a handful of agents who do prefer to report a budget that is somewhat in
between $0$s and their maximum budgets, we use numerical simulation to
calculate their best responses. We focus on a specific agent $i$.
Figure~\ref{fig:only} (Left) shows $i$'s real allocation curve.  When $i$'s
budget is below a minimum investment amount (denoted as $B_{min}$ in the
figure), her allocation is $0$. When $i$'s budget grows above a maximum amount
(denoted as $B_{max}$ in the figure), her allocation stays the same.  In
between, her allocation curve is \emph{approximately concave}: by investing
more, $i$ pushes up the price, so the marginal gain gets smaller and smaller.
The real allocation curve is unknown to us, so we resort to its approximation
in our equilibrium calculation.  We use piecewise straight lines to approximate
the real allocation curve, by sampling a few allocation values and then connect
them together. We call the sampled points the turning points in our piecewise
linear curve. To ensure that we end up with a concave curve, we use a linear
program to move the turning points slightly up or down. For example, let the
$y_i$ be the y-coordinates of the turning points ($y_0=0$). The linear program
is constructed as follows (the $y_i$ are constants and the variables are the
$y_i'$):

\begin{equation}
\begin{aligned}
& \text{min}
& & \delta \\
& \text{s.t.} & &  y_i-\delta \le y_i'\le y_i+\delta\\
    & & &  y_{i+1}'+y_{i-1}'\le 2y_{i}' \\
    & & &  y_{i}' \le y_{i+1}' \\
    & & &  y_{0}' = 0 \\
\end{aligned}
\end{equation}

It should be noted that in our approximation, we used a non-horizontal straight
line in between $0$ and $B_{min}$. This is fine because both before or after
the approximation, the agents' best strategies do not involve any budget
strictly in between $0$ and $B_{min}$.

With a concave allocation curve, the agent's utility function is also concave.
For example, when the budget is in between $B_{min}$ and $B_{max}$, the utility
curve is just the allocation curve times a constant (the agent's valuation),
then subtracts a linear term (the payment). According
to~\cite{Rosen1965:Existence}, for a $n$-person game with concave utility
function, a pure strategy Nash equilibrium always exist.  Our experiments show
that after the approximation, a pure strategy Nash equilibrium (a deterministic
budget profile) is very easy to find.

Our approximation introduces two sources of errors. First, we have the error
$\delta$ from the linear program. Then, as shown in
Figure~\ref{fig:only} (Right), there is error due to using a
straight line to approximate the real curve. Let $BL$ and $BR$ be the x-coordinates of
two adjacent turning points. The price at $BL$ is lower than the price at $BR$.
For any budget value $B$ in between $BL$ and $BR$, the allocation is in between
$B/p(BR)$ and $B/p(BL)$. The gap is at most $B\frac{p(BR)-p(BL)}{p(BL)p(BR)}$,
which is maximized when $B$ approaches the larger budget $BR$.  We go
through every adjacent pair of turning points to get the largest error.  Given these
two sources of errors, our computed Nash equilibrium is not an exact
equilibrium. Suppose when we calculate the best response for agent $i$, the
maximum error in terms of $i$'s utility is $\epsilon$, then we can only say this
agent's response is at most $2\epsilon$ away from the best response (we could
be underestimating the actual best response and overestimating the approximate
best response). 


The experiments are conducted using parameter $g=0.6$ and $K=\{\ceil{0.5n},
\ceil{0.6n}, \ceil{0.7n}, \ceil{0.8n}, \ceil{0.9n}, n\}$.
Table~\ref{tb:experiment} shows the decomposition of agents for different ICOs.
The data format is ``total number of agents = agents who report 0 based on
Proposition~\ref{prop:pullout} + agents who report the true maximum budget
based on Proposition~\ref{prop:maxout} + agents with nontrivial strategies.''
As shown in the table, the number of agents with nontrivial strategies are only
a handful.  Table~\ref{tb:experiment1} compares the equilibrium revenue under
the Gini mechanism (Gini Rev.) against the first-best optimal revenue (Opt.
Rev.).\footnote{Gini Rev. is higher than Opt. Rev. for Metronome due to
numerical error.} The unit is Ether.  For Polkadot, due to numerical error, one
agent keeps changing her budget back and forth while the other agents do not
change their budgets. The Gini revenue is only changed at the second digit
after the decimal point due to this agent's back and forth.  Err.  represents
the maximum utility error. That is, under our equilibrium, an agent's utility
is at most this value away from the best response utility.  The unit of error
is Ether, so it can be quite significant in its absolute value. For Gnosis,
$0.141$ Ether is worth about $40$ USD in November 2019.  Err./Budget is the
maximum ratio between the utility error and an agent's true maximum budget. For
Gnosis, this value is $0.145\%$. Errors for the other ICOs are significantly
better.

\begin{table}
    \caption{Agent Decomposition}
\centering
\begin{tabular}{l l}\label{tb:experiment}%
Raiden & 8574=1663+6907+4 \\
Metronome & 2884=384+2496+4 \\
Polkadot & 5910=1522+4377+11\\
GoNetwork & 5210=1351+3854+5\\
Gnosis &  1518=304+1209+5\\
\end{tabular}
\end{table}

\begin{table}
    \caption{Revenue and Error}
\centering
\begin{tabular}{ l c c c r }\label{tb:experiment1}
    & Gini Rev. & Opt. Rev. & Err. & Err./Budget \\
    \hline
    Raiden & 42177 & 42177 & 1.84e-3 & 2.63e-4\\
    Metronome & 8516 & 8515 & 6.85e-3 & 4.38e-4\\
    Polkadot & 104689 & 104735 & 9.75e-3 & 2.14e-4\\
    GoNetwork & 14352 & 14356 & 2.86e-3 & 2.60e-4\\
    Gnosis & 104117 & 104117 & 1.41e-1 & 1.45e-3\\
\end{tabular}
\end{table}
\newpage

\bibliographystyle{plain}
\bibliography{/home/mingyu/nixos/newmg.bib}
\end{document}